\documentclass[11pt]{article}

\usepackage{amsmath}
\usepackage{amssymb}
\usepackage{amscd}
\usepackage{amsthm}
\usepackage{flafter}
\usepackage{float}
\usepackage{url}
\usepackage{cite}
\usepackage{graphicx}

\usepackage{blkarray, bigstrut}
\usepackage{xparse}
\usepackage{amsfonts}
\usepackage[colorlinks=true]{hyperref}
\usepackage[usenames, dvipsnames]{color}
\definecolor{mypink}{cmyk}{0.99, 0, 0, 0.1}
\usepackage[dvipsnames]{xcolor}
\colorlet{mycolor}{black!20!blue!80!}
\hypersetup{urlcolor=blue, citecolor=blue}

\newtheorem{lemma}{Lemma}

\newtheorem{theorem}[lemma]{Theorem}

\newtheorem{remark}[lemma]{Remark}
\newtheorem{proposition}[lemma]{Proposition}
\newtheorem{example}[lemma]{Example}
\newtheorem{definition}[lemma]{Definition}

\renewcommand\ell{l} % do not use different script L

\title{Trade-Based LDPC Codes}

\author{Farzane Amirzade$^1$ \and Daniel Panario$^2$ \and Mohammad-Reza~Sadeghi$^1$ \\
  $^1$ Department of Mathematics and Computer Science, Amirkabir University of 
  Technology \\
  $^2$ School of Mathematics and Statistics, Carleton University \\
  \texttt{famirzade@gmail.com, daniel@math.carleton.ca, msadeghi@aut.ac.ir}
}
\begin{document}
\maketitle

% \title{Super-Simple Directed Group Divisible Designs With Block Size Four}
% 
% \author{Farzane Amirzade \and Daniel Panario \and Mohammad-Reza~Sadeghi}
% 
% \institute{Farzane Amirzade \at Department of Mathematics and Computer 
% Science, Amirkabir University of Technology, Tehran, Iran \\
% \email{famirzade@gmail.com} \and Daniel Panario \at
% School of Mathematics and Statistics of Carleton University, Ottawa, Canada\\
% \email{daniel@math.carleton.ca} \and Mohammad-Reza~Sadeghi \at 
% Department of Mathematics and Computer Science, Amirkabir University of 
% Technology, Tehran, Iran \\
% \email{msadeghi@aut.ac.ir} \and The authors were partially funded by the 
% Natural Sciences and Engineering Research Council (NSERC) of Canada.}

\begin{abstract}
LDPC codes based on multiple-edge protographs potentially have larger minimum distances compared to their counterparts, single-edge protographs. However, considering different features of their Tanner graph, such as short cycles, girth and other graphical structures, is harder than for Tanner graphs from single-edge protographs.  
In this paper, we provide a novel approach to construct the parity-check matrix of an LDPC code which is based on trades obtained from block
designs. We employ our method to construct two important categories of LDPC codes; quasi-cyclic (QC) LDPC and spatially-coupled LDPC (SC-LDPC) codes. 

We use those trade-based matrices to define base matrices of multiple-edge protographs. The construction of exponent matrices corresponding to these base matrices has less complexity compared to the ones proposed in the literature. We prove that these base matrices result in QC-LDPC codes with smaller lower bounds on the lifting
degree than existing ones. 

 There are three categories of SC-LDPC codes: periodic, time-invariant and time-varying. Constructing the parity-check matrix of the third one is more difficult because of the time dependency in  the parity-check matrix. We use a trade-based matrix to obtain the parity-check matrix of a time-varying SC-LDPC code in which each downwards row displacement  of the trade-based matrix yields syndrome matrices of a particular time. Combining the different row shifts the whole parity-check matrix is obtained. 
  
 Our proposed method to construct parity-check and base matrices from trade designs is applicable to any type of super-simple directed block designs. We apply our technique to directed designs  with smallest defining sets containing at least half of the blocks. To demonstrate the significance of our contribution, we provide a number of numerical and simulation results.
%\keywords{LDPC codes, Tanner graph, Girth, Multiple-edge LDPC codes, Time-varying spacially coupled LDPC codes, Super-simple directed designs, Trade designs, Defining sets.}
\end{abstract}

\section{Introduction}\label{I}
An LDPC code is a linear code (a vector space over a finite field) whose parity-check matrix is sparse. LDPC codes were discovered by Gallager in 1960 \cite{Gallager} and were mostly ignored during the 35 years that followed. One noticeable exception is the important work of Tanner in 1981 \cite{Tanner} in which the author generalized LDPC codes and introduced a graphical representation of LDPC codes now called Tanner graph.  It is a bipartite graph with two vertex sets: variable-nodes and check-nodes. The incidence matrix of the Tanner graph is the parity-check matrix of the code. A parity-check matrix with  column weight $c$ and row weight $r$ gives a $(c,r)$-regular LDPC code. If columns or rows have different weights, we obtain an irregular LDPC code. An important parameter of the code is the girth, that is, the length of the shortest cycles of its Tanner graph. Experimentally, it is known that the girth influences the code performance \cite{Lin2014}.

There are different approaches to construct LDPC codes. One of these construction methods is protograph-based.  A protograph is a small size bipartite graph whose adjacency matrix 
is taken as a base matrix $W$ of  protograph-based LDPC codes. There is a large body of work devoted to  protograph-based LDPC codes.  Protographs can result in quasi-cyclic (QC) LDPC codes and spatially-coupled (SC) LDPC codes. These codes perform very well under iterative decoding algorithms over the binary-input
AWGN channel \cite{2004,Lentmaier2010}.

QC-LDPC codes are divided into two categories, single-edge and multiple-edge. In fact, if the protograph is free of multiple-edges, then the QC-LDPC code is single-edge. Otherwise, it is multiple-edge. For each protograph-based LDPC code, two matrices are considered: a base matrix $W$ and an exponent matrix $B$. The elements of the base matrix of a single-edge QC-LDPC code are 0, 1. The base matrix of a multiple-edge QC-LDPC code contains also elements bigger than 1. According to \cite{Multiple,Sadeghi,LCOMM},  multiple-edge QC-LDPC codes are known to potentially have large minimum distances. However, obtaining the entries of exponent matrices of  multiple-edge QC-LDPC codes with a certain girth is more complex than those of single-edge QC-LDPC codes. Necessary and sufficient conditions to avoid short cycles such as 4-cycles and 6-cycles from the exponent matrix of a multiple-edge protograph are proposed in \cite{Sadeghi,LCOMM}. Those conditions depend on two difference matrices $D$ and $\Delta$ which are defind in Section \ref{IV}.

Recently,  the construction of protographs for SC-LDPC codes has attracted much attention. The protograph-based SC-LDPC code can be divided into three categories: periodic, time-invariant and time-varying. One of the methods to construct  time-invariant SC-LDPC codes is  similar to the construction of the parity-check matrix of QC-LDPC codes using a base matrix and an exponent matrix. Since there are several works in the literature regarding QC-LDPC codes, such similarity provides a natural way of obtaining a wide range of time-invariant SC-LDPC codes \cite{Sadeghi,Battag1,Battag2}. In fact,  one can use existing methods for the construction of  QC-LDPC codes to obtain time-invariant SC-LDPC codes. However, we cannot take advantage of those methods to find time-varying SC-LDPC codes.  In this case we face steep computational complexity requiremenrs when defining the exponent matrices. Indeed, the property of time dependency that is required in the parity-check matrix causes a raise in the computational complexity. Consequently,  the existing results in the literature about  time-varying SC-LDPC codes are much less in number than for the time-invariant ones. In addition, for the time-invariant codes defining a protograph is not a simple task. One approach to propose protographs of SC-LDPC codes is connecting together several SC-LDPC code chains which causes an improvement in the decoding performance \cite{Liu2015}. More recently,  a research group \cite{Truhachev} proved that properly choosing SC-LDPC code chains and connecting them at specific points give rise to an improvement in iterative decoding thresholds. However, all these methods are search-based and obtaining a method that works in some particular case not necessarily works in other cases.  

In this paper, we provide a novel approach to construct {\em parity-check matrices
of LDPC codes} with girth at least 6. Then, we use those matrices to propose
base matrices of multiple-edge QC-LDPC codes as well as {\em parity-check matrices} of time-varying SC-LDPC codes. To reach our goal we require special classes of combinatorial designs, such as $\emph{balanced incomplete}$ $\emph{block designs}$ (BIBD) as well as their related designs known as $\emph{trades}$. BIBDs have been employed to construct structured LDPC codes \cite{Multiple,9,10,11,12}. {\em The use of trades to construct codes is one of the contributions of this paper. In order to introduce} trades, in Section \ref{II}, we provide {\em some basic definitions from combinatorial design theory} and summarize related works in this area.

 Our proposed technique to construct a parity-check matrix for an LDPC code with girth at least 6 is based on the concept of trade in directed designs. We call the codes as $\emph{trade-based LDPC}$  $\emph{codes}$. Then, using trades and the obtained $\emph{trade-based matrices}$, we provide a new approach to construct a base matrix for a multiple-edge QC-LDPC code. We call the obtained LDPC codes as $\emph{trade-based multiple-edge QC-LDPC code}$. The main contribution of these base matrices, when compared to other constructions in the literature, is that they provide a considerable reduction in the computational complexity of the search algorithm used to identify the elements of the exponent matrix. In fact, we show that because of the non-existence of $2\times2$ submatrices of nonzero entries in the base matrix and because of their sparsity, by defining the entries of a small size submatrix of the exponent matrix we can identify all elements of the exponent matrix. Moreover, the sparseness property of our proposed base matrix causes a decrease in the number of short cycles, which are at the root of many harmful structures in the Tanner graph \cite{AMC}.  We also show that applying the methods in \cite{Sadeghi,LCOMM} to our proposed protographs results in a lower bound on the lifting degree which is  smaller than the ones in \cite{Sadeghi,LCOMM}. 

Another contribution of trade-based matrices, mentioned above, is {\em in the construction of parity-check matrices of time-varying SC-LDPC codes}. In fact, a trade-based matrix and different kinds of rearrangements of its rows give the parity-check matrix related to different times. Thus, without getting involved in protographs, base and exponent matrices, we can directly construct the parity-check matrix of a time-varying SC-LDPC code which is free of short cycles.  

The rest of the paper is organized as follows. Section \ref{II} presents basic notation and definitions. In Section \ref{III}, we propose our technique to construct trade-based LDPC codes. Our techniques are structural and applicable to all super-simple directed designs.  In Sections, \ref{IV} and \ref{V}, respectively, we explain how to use the trade-based matrix to construct base matrices of a multiple-edge QC-LDPC code and the parity-check matrix of a time-varying SC-LDPC code. We also elaborate on the merits of our proposed techniques.  In Section \ref{VI}, we summarize our results.

\section{Combinatorial designs preliminaries}\label{II}
In this section we  define several types of combinatorial designs that are important in our work. The Handbook of Combinatorial Designs \cite{CRC} is an encyclopedic presentation where these objects and their properties are given.

\begin{definition}
Let $v,k,\lambda,r$ be positive integer numbers and $V=\{x_1,\dots,x_v\}$ be a set of $v$ objects. A 2-$(v,k,\lambda)$-BIBD is a collection $\mathcal{B}$ of $k$-subsets of $V$, the $\emph{blocks}$ of the design, with the following properties: (1) each object $x_i$ appears in exactly $r$ blocks, and (2) every two objects appear together in exactly $\lambda$ blocks, where $r=\lambda\frac{v-1}{k-1}$. 
\end{definition}

The incidence matrix of a block design is taken as the parity-check matrix of an LDPC code. If $\lambda=1$, then the Tanner graph is 4-cycle free, otherwise it has girth 4 \cite{ChannelCodes}. 

\begin{definition}
A 2-$(v,k,\lambda)$ directed design $\mathcal{D}$, that we denote by 2-$(v,k,\lambda)$-DD, is a pair $(V,\mathcal{B})$, where $\mathcal{B}$ is a collection of $\emph{ordered}$ $k$-subsets of distinct elements of $V$ such that each $\emph{ordered}$ pair of distinct elements of $V$ appears in exactly $\lambda$ blocks.
\end{definition}
\begin{definition}
 Let $\mathcal{D}$ be a 2-$(v,k,\lambda)$-DD. The $\emph{defining\ set}$ of $\mathcal{D}$ is a subset of the blocks in $\mathcal{D}$ which occurs in no other 2-$(v,k,\lambda)$ directed design. 
 \end{definition}
 
  A 2-$(v,k,\lambda)$ design is $\emph{super-simple}$ if any two blocks intersects in at most two objects and it is simple if it has no repeated blocks.  When $k=3$, a super-simple design is just a simple design.  The concept of super-simple designs was introduced by Mullin and Gronau \cite{mullin}. There are well-known results for the existence of super-simple designs, especially for the existence of super-simple 2-$(v,k,\lambda)$-BIBDs. A super-simple $(v,k,\lambda)$-DD is equivalent to a super-simple 2-$(v,k,2\lambda)$-BIBD. In \cite{Quinn} it is shown that for each admissible value of $v$, there exists a simple 2-$(v,3,1)$-DD whose smallest defining sets have at least half of the blocks.  Moreover, for all $v\equiv1\ ({\rm mod}\ 3)$ there exist a super-simple 2-$(v,4,1)$-DD and a super-simple 2-$(v,4, 2)$-DD whose smallest defining sets have at least half of the blocks \cite{farzane1}, \cite{Boostan}.  Also, the existence of super-simple 2-$(v, 5, 1)$-DDs and their smallest defining sets for all $v\equiv1,5\ ({\rm mod}\ 10)$ is shown in \cite{farzane2}. 
 
 As suggested by A. P. Street, defining sets of directed designs are strongly related to $\emph{trades}$ \cite{street}. The concept of directed trade and defining set of directed designs are investigated in \cite{street}, \cite{Soltan}.  \begin{definition}
 A $(v,k,2)$ ${\emph directed\ trade}$ of volume $s$ consists of two
disjoint collections $T_1$ and $T_2$, each of $s$ blocks, such that every pair of distinct elements of $V$ is covered by precisely the same number of blocks of $T_1$ as of $T_2$. Such a directed trade is usually denoted by $T=T_1-T_2$.
\end{definition}

Given a directed trade $T=T_1-T_2$, blocks in
$T_1(T_2)$ are the $\emph {positive}$ $\emph{(negative)}$ blocks
of $T$. Let ${\mathcal{D}}=(V,\mathcal{B})$ be a directed design.
If $T_1\subseteq \mathcal{B}$, then $\mathcal{D}$ contains
the directed trade $T$ in which $T_2$ is built from $T_1$ by interchanging the order of some elements in blocks of $T_1$.  Hereafter, we use ${\emph trade}$, instead of directed trade.  To simplify considering trades we use a graph in which any vertex is associated to a block and two blocks $B_i, B_j$ are connected by an edge if and only if $B_i, B_j$ form a trade of volume 2. We call this a $\emph {graph of trades}$.

\begin{definition}
Let $K$ be a set of positive integer numbers smaller than or equal to $v$  and $\lambda$ be a positive integer.  A $(K,\lambda)$ ${\emph directed }$ ${\emph group\ divisible\ design}$ (DGDD) of type ${g_1}^{u_1}{g_2}^{u_2}\dots {g_{\ell}}^{u_{\ell}}$ with $ \sum_{i=1}^{\ell}{g_i}{u_i}=v$, is a triple $(V,\mathcal{G},\mathcal{B})$, where $V$ is a
$v$-set, $\mathcal{G}$ is a collection of subsets (groups), each of cardinality in $\{g_1,\dots, g_{\ell}\}$, which partition $V$ into $u_1$ groups of size $g_1,\dots,u_{\ell}$ groups of size $g_{\ell}$, and $\mathcal{B}$ is a collection of  blocks of $V$ such that if $B\in\mathcal{B}$, then $|B|\in K$. As before, each block is ordered and every pair of distinct elements of $V$ appears in precisely $\lambda$ blocks or one group but not in both.   If $\lambda=1$, $(K,1)$-DGDD is denoted by $K$-DGDD. 
\end{definition}

\begin{example}\label{trades}
Let us consider a super-simple 4-DGDD of type $2^4$ with $V=\{0,1,\dots,7\}$, 
groups $\mathcal{G}=\{\{0,1\}, \{2,3\}, \{4,5\}, \{6,7\}\}$ and blocks 
${\mathcal{B}}=\{(3,0,5,6), (7,5,0,2), (5,7,1,3), (6,4,3,1), (4,6,2,0),\\ 
 (1,2,6,5),  (2,1,7,4),
(0,3,4,7)\}.$ 
This directed design has $(v,k,\lambda)=(8,4,1)$  and contains 12 trades of volume 2 which are shown in Fig. \ref{DGDD2to4}. In this figure, each edge between two blocks  shows a trade of volume 2 and the trades related to the edges in boldface are in the following table.   In these trades the elements of $V$ in boldface are involved in the trades.	This design has also some trades of volume 4. For example, assuming $T_1=\{(3, 0,5,6), (7,5,0,2), (4,6,2,0), (1,2,6,5)\}\subseteq \mathcal{B}$ and $T_2=\{(3, 6,5,0), (7,2,0,5), (4,0,2,6), (1,5,6,2)\}$, $T=T_1-T_2$ gives an $(8,4,2)$ trade of volume 4. Moreover, with $T_1=\{(4,6,2,0), (6,4,3,1), (5,7,1,3), (7,5,0,2)\}$, $T_1\subseteq \mathcal{B}$ and $T_2=\{(6,4,0,2), (4,6,1,3), (7,5,3,1), (5,7,2,0)\}$, $T=T_1-T_2$ results in another $(8,4,2)$ trade of volume 4. 
\vspace{-0.5cm}
\begin{center}
\begin{figure}[ht]
\centering
\includegraphics[scale=0.4]{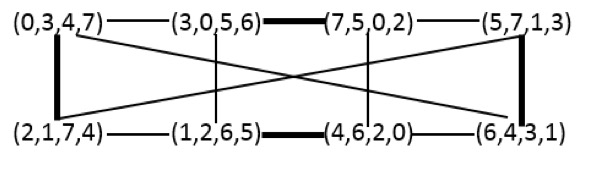}	
\caption{ The graph of trades of the design in Example \ref{trades}.}\label{DGDD2to4}
\vspace{-0.5cm}
\end{figure}
\end{center}
\vspace{-0.5cm}

\begin{table}[ht]
\begin{center}
\begin{tabular}{|l|l|} \hline
$\begin{array}{ccc}
T_1             &      & T_2 \\ \hline
				(3,{\bf 0,5},6)                    &      & (3, {\bf 5, 0}, 6)\\
				(7,{\bf 5,0},2)                     &      & (7, {\bf 0, 5}, 2) 
				\end{array}$& $\begin{array}{ccccccc}
				T_1&      & T_2  \\ \hline
				(5,7,{\bf 1,3})       &      & (5,7,{\bf 3,1})\\
				(6,4,{\bf 3,1})       &      &  (6,4,{\bf 1,3})
				\end{array}$\\ \hline
				$\begin{array}{ccc}
				T_1&      & T_2  \\ \hline
				(4,{\bf 6,2},0)       &      & (4, {\bf 2,6}, 0)\\
				(1,{\bf 2,6},5)       &      & (1, {\bf 6,2}, 5) 
				\end{array}$& $\begin{array}{ccccccc}
				T_1&      & T_2  \\ \hline
                (2,1,{\bf 7,4})       &      & (2,1,{\bf 4,7})\\
				(0,3,{\bf 4,7})       &      &  (0,3,{\bf 7,4})
				\end{array}$\\ \hline
			\end{tabular}
		\end{center}
	\end{table}
\vspace{-0.5cm}
\end{example}

Defining sets for directed designs are strongly related to trades. This relation
is illustrated by the following result.

\begin{proposition}\label{Prop1} 
	\cite{street}  Let $\mathcal{D}=(V,\mathcal{B})$ be a $2$-$(v,k,\lambda)$ directed design and  $S\subseteq
	\mathcal{B}$. Then $S$ is a defining set of $\mathcal{D}$ if and only if
	$S$ contains a block of every $(v,k,2)$ trade $T=T_1-T_2$
	such that $T_1$ is contained in $\mathcal{B}$.
\end{proposition}

Let $\mathcal{D}$ be a $(K,\lambda)$-DGDD. Each defining set $S$ of $\mathcal{D}$ contains at
least one block of each trade. In particular, if
$\mathcal{D}$ contains $m$ mutually disjoint trades, then
the smallest defining set of $\mathcal{D}$ must contain at least $m$
blocks. For example, the blocks of DGDD in Example \ref{trades} contain four mutually disjoint trades of volume 2 which are shown in the above table. As can be seen in Fig. \ref{DGDD2to4}, the boldfaced edges related to these trades  are distinct. We note that there could be different selections to find four distinct trades from 12 trades of volume 2 in the blocks of Example \ref{trades}, and in fact any four distinct edges in Fig. \ref{DGDD2to4} result in four distinct trades of volume 2. So, a defining set which contains one block of each of these trades is $S=\{(6,4,3,1),\ (5,7,1,3),\ (1,2,6,5),\ (0,3,4,7),(7,5,0,2)\}$. 
A special type of a trade that we also use to obtain defining sets is a
${\emph cyclical}$ ${\emph trade}$ which is defined next. 
\begin{definition}
A set of $s$ blocks $\{B_1,\dots,B_s\}$ forms a cyclical trade of volume $s$ if  each pair of consecutive blocks, $B_i,B_{i+1}$ for $1\leq i\leq s-1$ as well as $B_1,B_s$, forms a trade of volume 2, and so we have in total $s$ trades of volume 2. We denote a cyclical trade of volume $s$ by $CT_s$.
\end{definition}

 Two sets $CT_4=\{(3, 0,5,6), (7,5,0,2), (4,6,2,0), (1,2,6,5)\}$ and $CT_5=\{(3,0,5,6), (7,5,0,2),\\ (5,7,1,3), (2,1,7,4),(1,2,6,5)\}$ are, respectively, cyclical trades of volumes 4 and 5 in Example \ref{trades}.  Any $CT_s$ is equivalent to an $s$-cycle in the graph of trades.   In general, cycles in Fig. \ref{DGDD2to4} are of lengths 4, 5, 6, 7 and 8. As a result, the directed design in Example \ref{trades} contains $CT_4,CT_5,CT_6,CT_7$ and $CT_8$. 
 
 Therefore, if a directed design $\mathcal{D}$ contains a cyclical trade of volume $s$, then any defining set of $\mathcal{D}$ must contain at least $[\frac{s+1}{2}]$ blocks of that trade. It is important to note that a cyclical trade of volume $s$ does not necessarily give the positive blocks of a trade of volume $s$. For example, the cyclical trade $CT_4$ which is mentioned above gives the positive blocks of a trade of volume 4, because assuming $T_1=\{(3, 0,5,6), (7,5,0,2), (4,6,2,0), (1,2,6,5)\}\subseteq \mathcal{B}$ and $T_2=\{(3, 6,5,0), (7,2,0,5), (4,0,2,6), (1,5,6,2)\}$, $T=T_1-T_2$ gives an $(8,4,2)$ trade of volume 4. Whereas, $CT_5$ does not yield a trade of volume 5.   

 Here, we focus on super-simple directed designs whose smallest defining sets have at least half of the blocks. %We consider the parameter
%\begin{center}
	%$ f=\frac{\mbox{number of}\ k\mbox{-subsets in the smallest defining sets in } \mathcal{D}}{\mbox{number of }k\mbox{-subsets in }\mathcal{D}}.$
%\end{center}
%In other words, we are interested in designs for which the parameter $\large f$ is at least $\frac{1}{2}$. 
This property allows us to construct parity-check matrices of LDPC codes in which the number of rows and columns  is equal to or bigger than the number of blocks. Thus, the length of our proposed LDPC code is at least the number of blocks.

\section{Construction of trade-based LDPC codes}\label{III}
\raggedbottom In this section, we explain how to use directed block designs $(V,{\mathcal{B}})$ to obtain parity-check matrices of trade-based LDPC codes. Let $V=\{0,1,\dots,v-1\}$ be a $v$-set. Let us construct a $v\choose 2$ $ \times n$ binary matrix $A$ whose rows correspond to pairs $(x_i,x_j)$s, where $x_i<x_j\in\{0,1,\dots,v-1\}$, and its columns are labeled with $B_1,\dots, B_{n}$. If a pair $(x_i,x_j)$ belongs to a block $B_{\ell}$ which appears in a trade, then the element in the $(x_i,x_j)$-th row and the $\ell$-th column is 1, that is,  $A_{(x_i,x_j)\ell}=1$. Otherwise, it is zero. Then, we remove all-zero columns and all-zero rows of $A$ obtaining a binary matrix denoted by $C$. If the number of rows of $C$ is more than the number of columns, then $C^T$, the transpose of $C$, is the parity-check matrix of the trade-based LDPC code. Otherwise, $C$ is the parity-check matrix of the code. Since each block of length $k$ contains $k-1$ pairs $(x_i,x_j)$s, with $x_i<x_j$, which may appear in $\lambda$ trades, the number of 1s in each column of $A$ is at most $\lambda(k-1)$. Also, since any pair $(x_i,x_j)$, with $x_i<x_j$, occurs in $\lambda$ blocks, the number of 1s in each row of $A$ is at most $2\lambda$.

In the following theorem, we propose a close connection between the girth of  trade-based LDPC codes and the smallest volume of a cyclical trade. This relation is our main tool to investigate the cycle distribution of trade-based LDPC codes which are based on the cyclical trades of directed designs used to construct the codes.
\begin{theorem}\label{Prop2}
	A super-simple directed design has  a cyclical trade of volume $s$ if and only if the Tanner graph of the corresponding trade-based LDPC code has $2s$-cycles.
\end{theorem}
\begin{proof}
	Let $A$ be a matrix  corresponding to a super-simple directed design. If $B_{\ell},B_{k}\in \mathcal{B}$ are in a trade of volume 2, then $B_{\ell}\cap B_{k}=\{x_i,x_j\}$ and the $(x_i,x_j)$-th row of the matrix $A$ contains a 1 in column $B_{\ell}$ and a 1 in column $B_{k}$. 
	
	 Let $CT_s=\displaystyle \{B_1,B_2\dots, B_s\}$ be a cyclical trade of volume $s$. A cyclical trade of volume $s$ contains $s$ trades of volume 2 and according to the definition of  cyclical trade, we have $|B_i\cap B_{i+1}|=2, 1\leq i\leq s-1$, and $|B_1\cap B_s|=2$. Hence, a submatrix $E$ of $A$ on the columns $B_1,B_2\dots,B_s$  and the rows corresponding to the intersections of consecutive blocks of $CT_s$ is a matrix whose column weights and row weights are 2. Since the parity-check matrix of the trade-based code is obtained from $A$, the parity-check matrix contains $E$ or its transpose. Hence, $E$ or $E^T$ in the parity-check matrix corresponds to a cycle of length $2s$ in the Tanner graph of the code.
	 
	  For example, let $CT_3=\{B_1,B_2,B_3\}$ be a cyclical trade of volume 3, where $B_1\cap B_2=\{x_i,x_j\}$, $B_1\cap B_3=\{x_i,x_k\}$ and $B_2\cap B_3=\{x_j,x_k\}$. Then, the $B_1$-th column of $A$ contains a 1 in the $(x_i,x_j)$-th row and a 1 in the $(x_i,x_k)$-th row. Similarly,  the $B_2$-th column of $A$ contains a 1 in the $(x_i,x_j)$-th row and a 1 in the $(x_j,x_k)$-th row, and also the $B_3$-th column of $A$ contains a 1 in the $(x_i,x_k)$-th row and a 1 in the $(x_j,x_k)$-th row.  Therefore, the Tanner graph contains a 6-cycle corresponding to the path $A_{\left(x_i,x_j\right)B_1}\to A_{\left(x_i,x_j\right)B_2}\to A_{\left(x_j,x_k\right)B_2}\to A_{\left(x_j,x_k\right)B_3}\to A_{\left(x_i,x_k\right)B_3}\to  A_{\left(x_i,x_k\right)B_1}\to  A_{\left(x_i,x_j\right)B_1}$.

%\vspace{-0.2cm}
	  
%\begingroup\fontsize{8.5pt}{8.5pt}
%	\begin{align}
%\begin{blockarray}{cccccc}
	%$B_{1}$ & $B_{2}$  & $B_{3}$ & $B_{4}$ & $B_{5}$ & \\
	%\begin{block}{[ccccc]c}
%	${\bf 1}$ & ${\bf 1}$ & $0$ & $1$ & $0$ & $\left(x_i,x_j\right)$\\
%	${\bf 1}$ & $0$ & ${\bf 1}$ & $0$ & $1$ & $\left(x_i,x_k\right)$\\
%	 	$0$ & ${\bf 1}$ & ${\bf 1}$ & $1$ & $0$ & $\left(x_j,x_k\right)$\\
%\end{block}
%	\end{blockarray}\nonumber
%	\end{align}
%	\endgroup
%\vspace{-0.3cm}
	
Now, let the Tanner graph of a trade-based LDPC code obtained from a super-simple 
directed design have a $2s$-cycle. We prove that the corresponding directed design 
has a cyclical trade of volume $s$. Without loss of generality, we suppose that $C$ 
is the parity-check matrix of the trade-based LDPC code. The existence of a 
$2s$-cycle shows that the matrix $C$ contains an $s\times s$ submatrix $E$ such 
that each column and each row of $E$ contains at least two 1s. Without loss of 
generality, we can suppose that the columns of $E$ are $B_1,B_2\dots,B_s$ such 
that $B_i$ and $B_{i+1}$, as two variable-nodes, are in the $2s$-cycle between which 
there is one check-node. Therefore, for each $i\in\{1,\dots,s-1\}$, $B_i$ and 
$B_{i+1}$ appear in a trade of volume 2. Since $B_1$ and $B_s$ are connected to a 
common check-node,  $B_1$ and $B_s$ also appear in a trade of volume 2. According 
to the definition of the cyclical trade, we conclude that $B_1,B_2\dots,B_s$ produce 
a cyclical trade of volume $s$.\hfill
\end{proof}

From  Theorem \ref{Prop2} we conclude that the existence of cyclical trades of volume 3 results in 6-cycles in the Tanner graph of a trade-based LDPC code; otherwise the code has girth at least 8. Example \ref{Ex2} presents a trade-based  LDPC code from the DGDD of type $2^4$ in Example \ref{trades}. We use Theorem \ref{Prop2} to obtain the girth of its Tanner graph. 
\begin{example}\label{Ex2}
	We take all trades consisting of the blocks presented in Example \ref{trades}.  The matrix $C$ for this directed design has 12 rows and 8 columns. Thus, $C^T$, which is presented in the following matrix, is the parity-check matrix of a $(2,3)$-regular LDPC code: 
	\begingroup\fontsize{5.5pt}{5.5pt}
	\begin{align}\label{DGDD28}
C^T=\begin{blockarray}{ccccccccccccc}
	${\bf 02}$&${\bf 03}$  & ${\bf 12}$&${\bf 13}$& ${\bf 05}$ & ${\bf 17}$ & ${\bf 26}$ & ${\bf 34}$ & ${\bf 46}$ & ${\bf 47}$ & ${\bf 56}$ & ${\bf 57}$ &\\
	\begin{block}{[cccccccccccc]c}
	${\bf 1}$& .  & .& .&  ${\bf 1}$ & . & . & . & . &  .& . & ${\bf 1}$ & ${\bf 7502}$\\
		.& . &  . & ${\bf 1}$ &. & ${\bf 1}$ & . & . & . & .& . & ${\bf 1}$& ${\bf 5713}$\\
.& ${\bf 1}$  & .& .& ${\bf 1}$ & . & . & . &.&. &${\bf 1}$ & .& ${\bf 3056}$\\
.& .  & ${\bf 1}$ & .& . & . & ${\bf 1}$ & .& . & . &${\bf 1}$ & .  & ${\bf 1265}$\\
.& ${\bf 1}$  & . & .& .& . & .& ${\bf 1}$ & . & ${\bf 1}$ &. & .  & ${\bf 0347}$\\
	.& .  &  ${\bf 1}$ & .& .& ${\bf 1}$ & . & . & . & ${\bf 1}$ &. & . & ${\bf 2174}$\\
.& .  & . & ${\bf 1}$ & .& . & .&  ${\bf 1}$ & ${\bf 1}$ & . &. & .  & ${\bf 6431}$\\
	${\bf 1}$& . & . & . & .& . & ${\bf 1}$ & .& ${\bf 1}$ & . &. & . & ${\bf 4620}$\\
	\end{block}
	\end{blockarray}
	\end{align}
	\endgroup
	
	In this matrix, each column index $ij$ is equivalent to a pair $(i,j);\ i<j\in \{0,1,\dots,7\}$. Moreover, each row index $klmn$ is equivalent to a block $(k,l,m,n)\in \mathcal{B}$ and the sign $(.)$ denotes a zero entry. {\em As shown in Fig. \ref{DGDD2to4}, there is no triangle in the graph corresponding to the trades of this block design and therefore it is free of cyclical trades of volume 3. Thus, according to Theorem \ref{Prop2}, } the Tanner graph of the code has no 6-cycle. The existence of cyclical trades of volume 4 proves that the obtained LDPC code has girth 8.   
	
\end{example}

\begin{example}\label{trade2}
We consider a super-simple 4-DGDD of type $2^4$ with $V=\{0,1,\dots,7\}$, 
groups $\mathcal{G}=\{\{0,1\}, \{2,3\}, \{4,5\}, \{6,7\}\}$ and blocks 
${\mathcal{B}}=\{(0,3,6,5), (7,5,0,2), (5,7,3,1), (2,4,1,7), (4,6,2,0),\\
(1,2,5,6), (3, 0,7,4),  (6,1,4,3) \}.$ The pairs that appear in trades are 
$(0,2), (0,3), (1,4), (5, 6), (5,7)$. Therefore, the matrix $C$ for this design 
has 5 rows and 8 columns which is the parity-check matrix of an irregular LDPC 
code with column weights in $\{1,2\}$ and row weights 2. The matrix $C$ of this 
design, which is also used in Section \ref{IV}, is given next:

\begingroup\fontsize{5.5pt}{5.5pt}
	\begin{align}\label{C28}
	\begin{blockarray}{ccccccccc}
${\bf 0365}$	 & ${\bf 7502}$ & ${\bf 5731}$ & ${\bf 2417}$ & ${\bf 4620}$ & ${\bf 1256}$  & ${\bf 3074}$  &  ${\bf 6143}$ & \\
	\begin{block}{[cccccccc]c}
	.& ${\bf 1}$  &  . & . & ${\bf 1}$ & . & . &  . & ${\bf 02}$\\
	  ${\bf 1}$ & . &.  & . & . & . & ${\bf 1}$& . & ${\bf 03}$\\
      .& . & . & ${\bf 1}$ &.& . &. & ${\bf 1}$& ${\bf 14}$\\
      ${\bf 1}$& . & . & . & .& ${\bf 1}$ & .  & .  & ${\bf 56}$\\
      .  & ${\bf 1}$ & ${\bf 1}$ & . & . &. & . & .  & ${\bf 57}$\\
	\end{block}
	\end{blockarray}
	\end{align}
	\endgroup	
\end{example}

\begin{example}\label{Ex-col4}
The followings are the blocks of a $(20,5,1)$ DD in \cite{farzane2}:
{\small{$$\begin{array}{llll}
B_1=(1, 5, 9, 13, 17) & B_2=(2, 6, 10, 14, 17)& B_3=(3, 7, 11, 15, 17)& B_4=(4, 8, 12, 16, 17)\\
B_5=(17, 13, 12, 7, 2) & B_6=(17, 14, 11, 8, 1) & B_7=(17, 15, 10, 5, 4) & B_8=(17, 16, 9, 6, 3)\\
B_9=(2, 7, 9, 16, 18) & B_{10}= (3, 6, 12, 13, 18) & B_{11}=(4, 5, 11, 14, 18) & B_{12}=(1, 7, 12, 14, 19)\\
B_{13}=(18, 16, 12, 5, 1) & B_{14}=(18, 13, 9, 8, 4) & B_{15}=(18, 14, 10, 7, 3) & B_{16}=(19, 14, 9, 5, 2)\\
B_{17}=(3, 5, 10, 16, 19) & B_{18}=(4, 6, 9, 15, 19) & B_{19}=(1, 6, 11, 16, 0) & B_{20}=(2, 5, 12, 15, 0)\\
B_{21}=(19, 16, 11, 7, 4) & B_{22}=(19, 15, 12, 8, 3) & B_{23}=(0, 16, 10, 8, 2) & B_{24}=(0, 15, 9, 7, 1)\\
B_{25}=(4, 7, 10, 13, 0) & B_{26}=(3, 8, 9, 14, 0) &B_{27}= (2, 8, 11, 13, 19) & B_{28}=(1, 8, 10, 15, 18)\\
B_{29}=(0, 13, 11, 5, 3) & B_{30}= (0, 14, 12, 6, 4) & B_{31}= (19, 13, 10, 6, 1) & B_{32}= (18, 15, 11, 6, 2)
\end{array}$$}}
 As shown in the graph corresponding to trades of this design in Fig. \ref{V25K5}, each block appears in 4 trades of volume 2. The graph is free of triangles which proves the non-existence of cyclical trades of volume 3. The smallest cycle in the graph is 4 which corresponds to a $CT_4$. Using this design we can construct a matrix $C$ which has 32 rows and 64 columns. Therefore, $C^T$ is the parity-check matrix of a $(2,4)$-regular LDPC code with girth 8. We use the matrix $C$ of this code in Section \ref{V}.
 \vspace{-0.2cm}
 	\begin{center}
\begin{figure}
\centering
\includegraphics[scale=0.4]{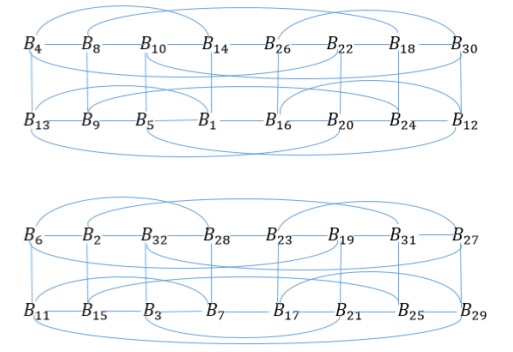}	
\caption{ {\em The graph of trades of the design in Example \ref{Ex-col4}.}}\label{V25K5}
\vspace{-0.2cm}
\end{figure}
\end{center}
\vspace{-0.8cm}
\end{example}

Our technique to construct a trade-based LDPC code is applicable to all kinds of super-simple directed designs. To clarify this technique, we apply it to super-simple DGDDs with block size 4 in the following examples. In these examples, we shall adopt the standard approach of using finite abelian groups to generate the set of blocks of any given directed design. For example, if $(V,\mathcal{G},\mathcal{B})$ is a DGDD, then instead of listing all the blocks, we give the element set $V$, plus a set of base blocks, and generate the other blocks using an additive group $G$. For $|G|=x$, the notation $(+t\ {\rm mod}\ x)$ means the base blocks should be developed by adding $0,\ t,\ 2t,\dots, x-t\ ({\rm mod}\ x)$ to them. Let us call this notation as a developer.   
\begin{example}\label{EX1}
	A super-simple 4-DGDD of type $2^7$  can be constructed on the set ${\mathbb{Z}}_{14}$ with groups $\{i,i+7\},\ i=0,1,\dots,6$, by developing
	the base blocks $(0,1,4,6)$ and  $(1,0,11,9)$ using $(+1\ {\rm mod}\ 14)$. We construct a binary matrix $A$ with $14\choose 2$ rows with indices $(x_i,x_j)$s, where $x_i,x_j\in\{0,1,\dots,13\}$, and $|{\mathcal{B}}|=28$ columns with indices $B_1,\dots, B_{28}$. The pairs $\displaystyle(0,1),\ (1,4),\ (4,6)\ (+1\ {\rm mod}\ 14)$,  appear in the trades. Let us take $B_1=(0, 1, 4, 6)$ as the first column index.  Then, $A_{(0,1)1}=1,\ A_{(1,4)1}=1$, $A_{(4,6)1}=1$ and all other entries of the fist column are zero. We continue this process to construct the binary matrix $A$ and remove all-zero rows and all-zero columns obtaining $C$. Since the number of rows and columns of $C$  are 42 and 28, respectively, $C^T$ is the parity-check matrix of the trade-based LDPC code. Since this design has no cyclical trade of volume 3, Theorem \ref{Prop2} proves the Tanner graph has girth at least 8. Moreover, the column weight and the row weight of the code is 2, 3, respectively.
\end{example}
\begin{example}\label{EX3}
	A super-simple $(4,2)$-DGDD of type $3^6$  can be constructed on the set ${\mathbb{Z}}_{18}$ with groups $\{i,i+6,i+12\};\ i=0,1,\dots,5$ and by developing the base blocks $(2,0,5,9),\ (7,10,0,2),\ (1,5,0,10),\\ (0,1,2,16), (11,4,1,0)$ by $(+1\ {\rm mod}\ 18)$; see \cite{Boostan}. We construct a binary matrix $A$ with $18\choose 2$ rows and 90 columns. Since the pairs appearing in the trades are $(0,2), (0,5), (5,9), (7,10), (0,10), (0,1)$ developed by  $(+1\ {\rm mod}\ 18)$, the number of rows of $C$ is 108. Since the number of rows of $C$ is more than the number of columns, $C^T$ is the parity-check matrix of this trade-based LDPC code.  This design contains cyclical trades of volume 3 such as $CT_3=\{(2,0,5,9),(7,10,0,2),(1,5,0,10)\}$. Hence, Theorem \ref{Prop2} proves the Tanner graph has girth 6. Since the column weights are 2 or 3 and the row weights are 2 or 3, the resulting code is irregular.
\end{example} 
\begin{example}\label{EX4}
We construct a super-simple $(4,3)$-DGDD of type $2^{10}$ on ${\mathbb{Z}}_{20}$ 
with groups $\{i,i+10\};\ i=0,1,\dots,9$ by developing by $(+1\ {\rm mod}\ 20)$ 
the base blocks:
\begin{align}
 & (7,0,2,13), (8,4,0,13), (2,0,11,14), (12,3,0,8), (2,0,11,14), \nonumber\\
 & (0,1,4,2), (1,5,6,0), (1,0,8,7), (5,8,0,2), (2,0,9,6). \nonumber
\end{align}
We construct a binary matrix $A$ with $20\choose 2$ rows and 180 columns. The 
pairs $(0,5), (0,7), (+1\ {\rm mod}\ 20)$ are not in trades, hence the number 
of rows of $C$ is 140. Thus, $C^T$ is the parity-check matrix of an irregular 
trade-based LDPC code whose row weights and column weights belong to $\{2,3,4\}$.
\end{example} 
\begin{example}\label{EX5D}
We provide a trade-based LDPC code with  a large girth Tanner graph. A super-simple 
5-DGDD of type $4^{16}$  can be constructed on the set ${\mathbb{Z}}_{64}$ with 
groups $\{i,i+16,i+32,i+48\};\ i=0,1,\dots,15$ by developing by $(+1\ {\rm mod}\ 64)$ the following 
base blocks; see \cite{farzane2}, $$(0,1,40,3,47), (1,0,25,62,18), (0,4,26,14,35), (4,0,42,54,33), 
(0,5,41,11,56), (5,0,28,58,13).$$
Every block appears in 4 trades; the smallest cyclical trade $CT_8$ has volume 8: 
\begin{align}
  & (0,5,41 11,56), (5,0,28,58,13), (2,7,43,13,58), (7,2,30,60,15),  \nonumber \\
  & (4,9,45,15,60), (9,4,32,62,17), (6,11,47,17,62), (52,47,11,41,60).  \nonumber
\end{align} 
We construct a binary matrix $A$ with $64\choose 2$ rows and 384 columns. Since 
the number of pairs that appear in trades is 768, $C^T$ is the parity-check 
matrix of a $(2,4)$-regular trade-based LDPC code  with girth 16.
\end{example} 
\section{Trade-based multiple-edge QC-LDPC codes}\label{IV}
A base matrix $W$ is associated to each protograph-based QC-LDPC code whose elements identify the type of the protograph. A base matrix with entries 0 and 1 only results in a single-edge protograph. A base matrix containing $l\geq2$ corresponds to a multiple-edge QC-LDPC code. First, we explain how to construct the parity-check matrix of QC-LDPC codes.
\begin{definition}
	Given a lifting degree $N$, an exponent matrix $B$ corresponding to an $m\times n$ base matrix $W$ is defined as follows.  If $W_{ij}=\ell$ is a positive integer, then the $ij$-th entry of $B$ is an $\ell$-set, that is,  $B_{ij}=\{b^{1}_{ij},b^{2}_{ij},\ldots, b^{l}_{ij}\}$, where $b^{r}_{ij}\in \lbrace 0,1,\ldots,N-1\rbrace$ and $\;b^{r}_{ij}\neq b^{r^\prime}_{ij}$ for $\;1\leq r<r^\prime\leq l$. If $W_{ij}=0$, then the $ij$-th entry of $B$ is an empty set, that is, $B_{ij}=\emptyset$.
	%$$\small{B=\left[\begin{array}{cccc}
		%B_{00}&B_{01}&\cdots &B_{0(n-1)}\\
		%B_{10}&B_{11}&\cdots &B_{1(n-1)}\\
		%\vdots &\vdots &\ddots &\vdots \\
	%	B_{(m-1)0}&B_{(m-1)1}&\cdots &B_{(m-1)(n-1)}\\
	%	\end{array}\right].}$$

If the $ij$-th entry of $B$ is an $\ell$-set, then it is replaced by an $N\times N$ matrix 
	$$ H_{ij}=I^{b^{1}_{ij}}+I^{b^{2}_{ij}}+\cdots +I^{b^{l}_{ij}},$$ where $I^{b^{r}_{ij}}$ is an $N\times N$ circulant permutation matrix (CPM) whose top row has a single 1 at the $b^{r}_{ij}$-th position. The other entries of its first row are zero. The $s$ right cyclic shifts of the first row of the CPM yield its $s$-th row. If $B_{ij}=\emptyset$, then it is substituted by an $N\times N$ zero matrix. The resulting matrix is the parity-check matrix whose null space gives a QC-LDPC code \cite{AMC}.
\end{definition}

A necessary and sufficient condition given in \cite{2004} for the existence of $2k$-cycles  in  the Tanner graph of  multiple-edge QC-LDPC codes is, for $n_{k}=n_{0}$, 
\begingroup\fontsize{11pt}{11pt}\begin{align}\label{Equation2}
\sum_{i=0}^{k-1} \left( b^{r^{}_i}_{m_{i}n_{i}} - b^{r^{\prime}_i}_{m_{i}n_{i+1}} \right)\equiv 0\  ({\rm mod}\ N) , 
\end{align}\endgroup
where if $n_{i} = n_{i+1}$, then $r^{}_{i}\neq r^{\prime}_{i}$, and also if $m_{i} =m_{i+1}$, then $ r^{\prime}_{i}\neq r^{}_{i+1}$. Here $b^{r^{}_i}_{m_{i}n_{i}}$ is the $r^{}_i$-th entry of $B_{m_{i}n_{i}}$ and $B_{m_{i}n_{i}}$ is the $(m_i n_i)$-th entry of $B$.

In order to simplify checking for $2k$-cycles, we need the difference matrices $D$ and $\Delta$ given in \cite{LCOMM}. In addition to these, we need a necessary and sufficient condition to have a girth-6 multiple-edge QC-LDPC code also given in \cite{LCOMM}.
\begin{definition}\label{DefD}
\cite{LCOMM} A difference matrix $D$	corresponding to $B$ is a matrix in which the $ij$-th element is obtained by subtracting every two entries of $B_{ij}$ modulo $N$. If $|B_{ij}|>1$, then there exist $2{|B_{ij}|\choose 2}$ choices to obtain $D_{ij}$ and for the chosen $x,y\in B_{ij}$, there are the subtractions $(x-y)$ and $(y-x)$ modulo $N$. Thus, if $|B_{ij}|>1$, then $|D_{ij}|=2{|B_{ij}|\choose 2}$. If $|B_{ij}|=1$, then $D_{ij}=\emptyset$, and if $B_{ij}=\emptyset$, then $D_{ij}=\emptyset$. 
 	\end{definition}
 	
 	\begin{definition}\label{DefDelta}
\cite{LCOMM} A difference matrix $\Delta$ of size ${m\choose2}\times n$ corresponding to an $m\times n$ matrix $B$ is a matrix in which each two rows of $B$ with indices $i$ and $i'$, $i<i'$ give the $(i,i')$-th row of $\Delta$ in the following way. If  $B_{ij}$ and  $B_{i'j}$ are two entries of the $j$-th column of $B$, then $\Delta_{(i,i')j}$ of size $|B_{ij}|\times|B_{i'j}|$ is obtained by subtracting each entry of $B_{i'j}$ from each entry of $B_{ij}$. If $B_{ij}$ or $B_{i'j}$  is $\emptyset$, then $\Delta_{(i,i')j}=\emptyset$.
 \end{definition}
 
 \begin{theorem}\label{Theg6}
\cite{LCOMM} A multiple-edge QC-LDPC code is free of 4-cycles if and only if each row of $D$, each column of $D$ and each row of $\Delta$ are free of repeated elements.
\end{theorem}

Now, we use the matrix $C$ or $C^T$ proposed in Section \ref{III} to construct a base matrix for a QC-LDPC code. 
\begin{definition}\label{P}
A matrix $P = [C_1|C_2|\cdots|C_r]$ of the maximum size free of any all-one submatrix of size 2 is defined as follows. First, we call the matrix $C$ or $C^T$ as $C_1$.  By downwards row shifting in $C_1$ other matrix named as $C_2$ is obtained with the same properties as $C_1$ in terms of the cycle distribution. We try on different  row displacements to find a $C_2$ such that concatenating $C_1$ and $C_2$ does not cause a $2\times2$ all-one submatrix in $[C_1|C_2]$. We continue this process to find other  $C_i$s and $P$ of the maximum size.
\end{definition}

\begin{remark}
It should be noted that the downwards row shifting mentioned in Definition \ref{P} can be replaced with different permutaions on row's positions. For any row permutation that we choose to obtain $C_i$s, we should consider the property that concatenating each two $C_i$ and $C_j$ does not cause a $2\times2$ all-one submatrix in $[C_i|C_j]$.
\end{remark}

Next, we explain how to construct the base matrix of  trade-based QC-LDPC codes using the matrix $P$.  Suppose the base matrix corresponding to $C_1$ is denoted by $W_1$ which is obtained by converting all 1s of $C_1$ to integers $l\geq1$  . In order to define the base matrix $W$, we can consider the same process that we have to obtain $P$.  First, we identify the entries of $W_1$, then we take $W=[W_1|\cdots|W_r]$ such that each $W_i$ is a downwards row displacement of $W_1$.
\begin{example}\label{EXP5}
As explained in Example \ref{trade2}, from a super-simple 4-DGDD of type $2^4$ we construct a $5\times8$  matrix $C$ presented in Equation (\ref{C28}) which yields an irregular LDPC code with column weights in $\{1,2\}$ and row weight 2. In the following we show how we use the matrix $C$ to construct  a base matrix of a $(3,24)$-regular multiple-edge QC-LDPC code. Taking $C$ as $C_1$, we construct a matrix $P$ of the maximum size free of any all-one submatrix of size 2. The largest size of a matrix that we obtain from $C$  is $P=[C_1|C_2|C_3|C_4|C_5]$. Since the column weight is supposed to be 3, all 1s in the last six columns of $C_1$ are replaced with 3. In the first two columns of $C_1$ we have two 1s which are replaced by 1 and 2 which sum to 3. Thus, by replacing all 1s of $C_1$ with integers $1,2,3$ we define $W_1$. We observe that the columns of $W_1$ have weight 3. Then, by different downwards row displacements of $W_1$ we construct other $W_i$s such that concatenating them gives a base matrix in which the weights of each row sum to 24. We provide an instance of such base matrix:
	
{\small{$$W=\left[\begin{array}{c|c|c|c|c}
	0\, 1\,  0\, 0\, 3\, 0\, 0\, 0 &  0\, 2\, 3\, 0\, 0\, 0\, 0\, 0\, & 2\, 0\,  0\, 0\, 0\, 3\, 0\,  0\, & 0\, 0\,  0\, 3\, 0\, 0\,  0\, 3\, &  1\, 0\,  0\, 0\, 0\, 0\, 3\, 0\\
	
		1\, 0\,  0\, 0\, 0\, 0\, 3\, 0 & 0\,1\,  0\, 0\, 3\, 0\, 0\, 0 & 0\, 2\,  3\, 0\, 0\, 0\, 0\, 0 & 2\, 0\,  0\, 0\, 0\, 3\, 0\, 0 & 0\, 0\,  0\, 3\, 0\, 0\, 0\, 3\\
	
	0\, 0\,  0\, 3\, 0\, 0\, 0\, 3 & 1\, 0\,  0\, 0\, 0\, 0\, 3\, 0 & 0\, 1\,  0\, 0\, 3\, 0\, 0\, 0 & 0\, 2\,  3\, 0\, 0\, 0\, 0\, 0 & 2\, 0\,  0\, 0\, 0\, 3\, 0\, 0 \\
	
	2\, 0\,  0\, 0\, 0\, 3\, 0\, 0 & 0\, 0\,  0\, 3\, 0\, 0\, 0\, 3 & 1\, 0\,  0\, 0\, 0\, 0\, 3\, 0 & 0\, 1\,  0\, 0\, 3\, 0\, 0\, 0 & 0\, 2\,  3\, 0\, 0\, 0\, 0\, 0\\
	
	0\, 2\,  3\, 0\, 0\, 0\, 0\, 0 & 2\, 0\,  0\, 0\, 0\, 3\, 0\, 0 & 0\, 0\,  0\, 3\, 0\, 0\, 0\, 3 & 1\, 0\,  0\, 0\, 0\, 0\, 3\, 0 & 0\, 1\,  0\, 0\, 3\, 0\, 0\, 0\\
	
	\end{array}\right].$$}} 
\end{example}

\begin{example}\label{EXP1}
If we take $C^T$ in Example \ref{Ex2} as $C_1$, then the matrix $P$ of the maximum size free of any all-one submatrix of size 2 is $P = [C_1]$.  Replacing all 1s by 2 we obtain a base matrix for a $(4,6)$-regular multiple-edge QC-LDPC code. 
\end{example}

\begin{remark}\label{ExpME}
As explained above, the base matrix of our approach has the form $W=[W_1|\cdots|W_r]$. Suppose the exponent matrix associated to $W_1$ is denoted by $B_1$. In order to define the exponent matrix $B$, we can consider the same process that we have to obtain $W$.  First, we identify the entries of $B_1$, then we take $B=[B_1|\cdots|B_r]$ such that each $B_i$ is the  row shift of $B_1$ as $W_i$ is the row shift of $W_1$. Using this method, the size of the search space to define the exponent matrix is considerably reduced. In fact, instead of identifying all entries of $B$ we only determine the entries of $B_1$.
\end{remark}

\begin{example}\label{ExamP5}
As shown in Example \ref{EXP5}, replacing all 1s in $P=[C_1|\cdots|C_5]$  with 1, 2, 3 we obtain the base matrix of a $(3,24)$-regular multiple-edge QC-LDPC code. 	According to Remark \ref{ExpME}, to define the exponent matrix $B$, we can execute the same process that we have done to obtain $W$.   First, we identify the entries of $B_1$, assuming that $N=41$ is the lifting degree:
{\small $$B_1={\left[\begin{array}{cccccccc}
	\emptyset & \{0\} &  \emptyset & \emptyset & \{0,1,3\} & \emptyset & \emptyset & \emptyset\\
		\{0\} & \emptyset &  \emptyset & \emptyset & \emptyset & \emptyset & \{0,4,9\} & \emptyset \\
	\emptyset & \emptyset &  \emptyset & \{0,6,13\} & \emptyset & \emptyset & \emptyset & \{0,8,22\}  \\
	\{7,27\} & \emptyset &  \emptyset & \emptyset & \emptyset & \{0,10,25\} & \emptyset & \emptyset \\
	\emptyset & \{19,36\} &  \{6,24,36\} & \emptyset & \emptyset & \emptyset & \emptyset & \emptyset\\
	\end{array}\right].}$$ }
 	
The process of finding the entries of $B_1$ will be explained in Example \ref{EX13}. In fact, the entries are defined such that they do not result in 4-cycles in the Tanner graph.  Then, we take $B=[B_1|\cdots|B_5]$ such that each $B_i$ is a  row-shift downwards displacement of $B_1$ and is associated to $W_i$. Since the number of rows of $B_1$ is 5 and $r=5$ (the number of $B_i$s), all rows of $B_1$ appear in each row of $B$. Applying the method explained in Remark \ref{ExpME},  by identifying only 24 entries  we can construct an exponent matrix of a $(3,24)$-regular QC-LDPC code. Thus,  the size of the search space to obtain the entries of $B$ is reduced from $N^{120}$ to $N^{24}$. 
	\end{example}

In the following, we provide lower bounds on the lifting degree of trade-based multiple-edge QC-LDPC codes with girth 6. We show that these lower bounds are less than the ones proposed in the literature.

\begin{theorem}\label{ME-g6}
\cite{LCOMM}	Let $W$ be an $m\times n$ base matrix of a multiple-edge QC-LDPC code. We define three parameters as follows:
	
	  ${\mathcal X} =\max\{\sum_{j=0}^{n-1}{W_{ij}\choose2}\colon\ 0\leq i\leq m-1\}$, 
	
	 ${\mathcal Y}=\max\{\sum_{i=0}^{m-1}{W_{ij}\choose2}\colon\ 0\leq j\leq n-1\}$ and
	 
	 ${\mathcal Z}=\max\{\sum_{j=0}^{n-1}W_{ij}\times W_{i'j}\colon\ i\neq i';\ 0\leq i,i'\leq m-1\}$. 

\noindent The lower bound on the lifting degree of a multiple-edge QC-LDPC code with girth at least 6 is  $N=\max\{2{\mathcal X},2{\mathcal Y},{\mathcal Z}\}$.
\end{theorem}

\begin{example}
Let us consider a base matrix 
$W=\left[\begin{array}{cccc}
	2 & 1 & 2 & 0\\
	2 & 2 & 3 & 1\\
	1 & 1 & 2 & 3
\end{array}\right]$. 
We have ${\mathcal X}=\max\{2,5,4\}=5,\ {\mathcal Y}=\max\{2,1,5,3\}=5$ and 
${\mathcal Z}=\max\{12,7,13\}=13$. Thus, $N\geq\max\{10,10,13\}=13$. 
\end{example}
\begin{theorem}\label{TD-ME-g6}
If $W$ is an $m\times n$ base matrix of a trade-based multiple-edge QC-LDPC code with girth at least 6, then $N\geq\max\{2{\mathcal X},2{\mathcal Y}\}$.	
\end{theorem}
\begin{proof}
	We prove that the  lower bound on the lifting degree of a trade-based multiple-edge QC-LDPC code with girth at least 6 is bigger than ${\mathcal Z}$. Since the base matrix of our desired code is obtained from a matrix $P$ which is free of an all-one submatrix of size 2, for each two rows with indices $i,i'\in\{0,\dots,m-1\}$ there is only one column $j\in\{0,\dots,n-1\}$ in which $W_{ij}\neq0$ and  $W_{i'j}\neq0$. Hence, the parameter ${\mathcal Z}$ is equal to the multiplication of two positive integers belonging to a column of $W$. Without loss of generality, we suppose that  these integers are in column $\ell$ and rows $r$ and $s$  of $W$. Therefore, we have ${\mathcal Z}=W_{r\ell}\times W_{s\ell}$. Since the $\ell$-th column may contain other positive integers we have $$2{\mathcal Y}\geq 2\sum_{i=0}^{m-1}{W_{i\ell}\choose2}\geq2\left({W_{r\ell}\choose2}+{W_{s\ell}\choose2}\right).$$ 
	
	It can be easily seen that for two positive integers $a,b$ which at least one of them is bigger than 1 we have $2({a\choose2}+{b\choose2})\geq a\times b$. Consequently, $$2\left({W_{r\ell}\choose2}+{W_{s\ell}\choose2}\right)\geq W_{r\ell}\times W_{s\ell}={\mathcal Z}.$$
	This implies $2{\mathcal Y}\geq{\mathcal Z}.$ 
\hfill
	\end{proof}
	
In the following examples, we elaborate the advantages of our approach in reducing the size of the search space when identifying the entries of an exponent matrix. In fact, we show that because of the non-existence of $2\times2$ submatrices of nonzero entries in the base matrix and because of its sparsity, by defining the entries of a small size submatrix of the exponent matrix we can identify all elements of the exponent matrix.
\begin{example}\label{EX13}
 Consider a base matrix which yields a $(3,24)$-regular multiple-edge QC-LDPC code with girth 6. If  the base matrix is a $2\times16$ matrix with elements 1, 2, then applying Theorem \ref{ME-g6} to this base matrix proves  ${\mathcal X}=8,\ {\mathcal Y}=1$ and ${\mathcal Z}=32$ which results in $N\geq\max\{16,2,32\}=32$. To obtain such matrix, we need to define all entries of the exponent matrix. 	Hence, the size of the search space is $N^{48}$.

As shown in Example \ref{EXP5}, replacing all 1s in $P=[C_1|\cdots|C_5]$  with 1, 2, 3 we obtain the base matrix resulting in a $(3,24)$-regular multiple-edge QC-LDPC code. If this code has girth 6, then according to Theorem \ref{TD-ME-g6}, the exponent matrix of this base matrix with  ${\mathcal X}=20,\ {\mathcal Y}=3$  has $N\geq\max\{40,6\}=40$.  Now to have girth 6, we should apply the necessary and sufficient conditions of Theorem \ref{Theg6} to the exponent matrix $B$ of Example \ref{ExamP5} which is constructed according to Remark \ref{ExpME}. Thus, we need the difference matrices corresponding to $B$. We take the entries of $B_1$ as $x_i=0$ for $i=1,2,5,6,9,12,17$ and $x_3=1,x_4=3,x_7=4,x_8=9,x_{10}=6,x_{11}=13,x_{13}=8,x_{14}=22,x_{15}=7,x_{16}=27,
x_{18}=10,x_{19}=25,x_{20}=19,x_{21}=x_{24}=36,x_{22}=6,x_{23}=24$.
 If we define $d_{i,j}$ as $x_i-x_j\pmod N$ and $-d_{i,j}$ as $N-d_{i,j}$ then we conclude that  since the set $\{\pm d_{2,3}, \pm d_{2,4}, \pm d_{3,4},\ \pm d_{6,7}, \pm d_{6,8}, \pm d_{7,8},\ \pm d_{9,10}, \pm d_{9,11}, \pm d_{10,11},\ \pm d_{12,13}, \pm d_{12,14}, \pm d_{13,14},\ \pm d_{15,16},\\ 
\pm d_{17,18}, \pm d_{17,19}, \pm d_{18,19}, \pm d_{20,21},
 \pm d_{22,23}, \pm d_{22,24}, \pm d_{23,24}\}$ is free of repeated entries and each of the sets $\{\pm d_{5,15},\pm d_{5,16}\}$ and $\{\pm d_{1,20},\pm d_{1,21}\}$ have 4 distinct elements,  the difference matrices of $B$ hold the conditions of  Theorem \ref{Theg6}. Thus,  the matrix $B$ yields a girth-6 code.  Using this method,  by identifying only 24 entries shown in Example \ref{ExamP5}, we construct an exponent matrix of a $(3,24)$-regular QC-LDPC code. Thus, the size of the search space is reduced from $N^{48}$ to $N^{24}$ which is a considerable reduction. 
	\end{example}
	
\begin{example}
Consider a base matrix which yields a $(4,6)$-regular multiple-edge QC-LDPC code with girth 6. If  {\small{$W=\left[\begin{array}{ccc}
2 & 2 & 2\\
2 & 2 & 2
\end{array} \right]$}}, then applying Theorem \ref{ME-g6} to this base matrix gives ${\mathcal X}=3,\ {\mathcal Y}=2$ and ${\mathcal Z}=12$ which results in $N\geq\max\{6,4,12\}=12$. An exponent matrix of this code with $N=12$ which satisfies the conditions of Theorem \ref{ME-g6} and gives a Tanner graph with girth 6 is  
{\small{$B=\left[\begin{array}{ccc}
(0,1) & (0,3) & (0,7)\\
(0,7) & (1,5) & (3,4)
\end{array} \right]$}}.
{\em To obtain such matrix, we need to define all entries of the exponent matrix through an exhaustive search algorithm}. Thus, the size of the search space is $N^{12}$. Whereas, using the base matrix of Example \ref{EXP1}, which is obtained by replacing 1s of the matrix (\ref{DGDD28}) by 2, we have ${\mathcal X}=3,\ {\mathcal Y}=2$ and then, according to Theorem \ref{TD-ME-g6}, $N\geq\max\{6,4\}=6$. Moreover, as a result of sparsity of the base matrix, we can determine the entries of $B$ only by identifying the elements of the first row:
{\small{$$B=\left[
\begin{array}{cccccccccccc}
	\{0,1\} & \emptyset  & \emptyset & \emptyset &  \{0,2\} & \emptyset & \emptyset & \emptyset & \emptyset &  \emptyset & \emptyset & \{3,6\} \\
		\emptyset & \emptyset &  \emptyset & \{0,1\} & \emptyset & \{3,6\} & \emptyset & \emptyset & \emptyset & \emptyset & \emptyset & \{0,2\}\\
\emptyset & \{0,1\}  & \emptyset & \emptyset & \{3,6\} & \emptyset & \emptyset & \emptyset & \emptyset & \emptyset & \{0,2\} & \emptyset\\
\emptyset & \emptyset  & \{3,6\} & \emptyset & \emptyset & \emptyset & \{0,2\} & \emptyset & \emptyset & \emptyset & \{0,1\} & \emptyset  \\
\emptyset & \{3,6\}  & \emptyset & \emptyset & \emptyset & \emptyset & \emptyset & \{0,2\} & \emptyset & \{0,1\} & \emptyset & \emptyset  \\
	\emptyset & \emptyset  &  \{0,1\} & \emptyset & \emptyset & \{0,2\} & \emptyset & \emptyset & \emptyset & \{3,6\} & \emptyset & \emptyset \\
\emptyset & \emptyset  & \emptyset & \{3,6\} & \emptyset & \emptyset & \emptyset&  \{0,1\} & \{0,2\} & \emptyset. & \emptyset & \emptyset  \\
	\{0,2\} & \emptyset & \emptyset & \emptyset & \emptyset & \emptyset & \{0,1\} & \emptyset & \{3,6\} & \emptyset & \emptyset & \emptyset \\
	\end{array}\right].
	$$}}
	
As can be seen, the elements of all rows except for the first one are a permutation 
of the first row. Thus, using our approach we reduce the computational complexity 
in the search algorithm from $N^{12}$ to $N^6$. Now, we show that the exponent 
matrix $B$ yields a Tanner graph with girth 6. We take $x_1=x_3=0,x_2=1,x_4=2,x_5=3$ 
and $x_6=6$ and assume $d_{i,j}$ is $x_i-x_j\pmod N$ and $-d_{i,j}$ is $N-d_{i,j}$. 
Since the sets 
\begin{align}
 & \{\pm d_{1,2},\pm d_{3,4},\pm d_{5,6}\}, \{ d_{1,5}, d_{1,6}, d_{2,5},d_{2,6}\}, 
    \{ d_{1,3}, d_{1,4}, d_{2,3}, d_{2,4}\}, \{ d_{5,1}, d_{6,1}, d_{5,2},d_{6,2}\},  \nonumber \\
 &   \{ d_{3,1}, d_{4,1}, d_{3,2}, d_{4,2}\}, 
   \{ d_{3,5}, d_{3,6}, d_{4,5}, d_{4,6}\}, \{ d_{5,3}, d_{6,3}, d_{5,4},d_{6,4}\}  \nonumber 
\end{align}
are free of repeated elements, the matrix $B$ satisfies the necessary and sufficient 
conditions of Theorem \ref{Theg6}. Thus, $B$ gives a Tanner graph with girth 6.

The performance curves of these two codes, one with $N=12$ and the other with $N=7$, 
decoded using the sum-product algorithm with 50 iterations are shown in 
Fig. \ref{Figure}. {\em As can be seen the trade-based code has better performance 
which is a result of longer code length and sparsity of the protograph and thus less 
cycle distribution.}
\end{example}	
\vspace{-0.5cm}

\begin{center}
\begin{figure}
\centering
\includegraphics[scale=.45]{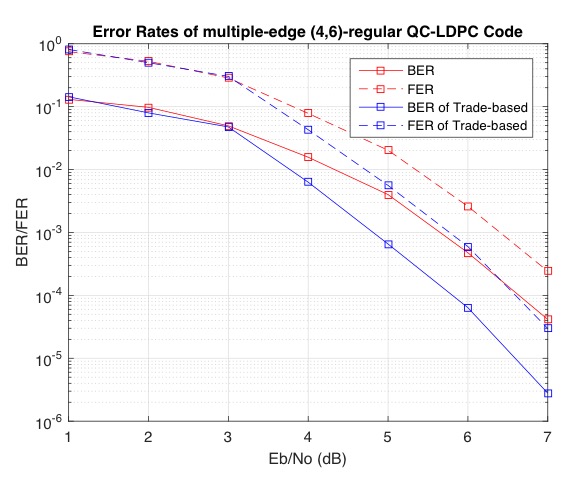}	
\caption{The comparison of the  performance curves of two $(4,6)$-regular 
multiple-edge QC-LDPC codes. One of these codes with a better performance is 
trade-based, the other has a fully-connected protograph.}\label{Figure}
\end{figure}
\end{center}
\vspace{-0.5cm}

In the following, we present a lower bound on the lifting degree of our proposed multiple-edge 
QC-LDPC code with girth 8. It is known that $W_{ij}\geq3$ causes 6-cycles 
\cite{Sadeghi}. Therefore, to have a base matrix for a girth-8 code, all 1s of 
the matrix $P$ have to be replaced by an $\ell=1,2$. In \cite{Sadeghi}, a lower 
bound on the lifting degree of a multiple-edge QC-LDPC code with girth 8 has been 
given. In order to give a better lower bound for our proposed base matrix we need 
the following definition from \cite{Sadeghi}.
\begin{definition}\label{Definition3}
\cite{Sadeghi}	Given an $m\times n$ exponent matrix $B$, a difference matrix $\Delta$ of size ${m\choose2}\times n$ corresponding to $B$ is a matrix in which a row with a row index $(i,i')$ is associated to each two rows of $B$ with indices $i$ and $i'$, $i<i'$. If $B_{ij}$ and $B_{i'j}$ are two entries   of the $j$-th column of $B$ with cardinality at least 1, then  a vector $\vec{\Delta}_{(i,i')j}$ of size $|B_{ij}|\times|B_{i'j}|$ is obtained from subtracting each entry of $B_{i'j}$ from  each entry of $B_{ij}$. Moreover, if  $B_{ij}$ or $B_{i'j}$ is $\emptyset$, then $\vec{\Delta}_{(i,i')j}=\emptyset$ . 
\end{definition}

For each $2\times2$ submatrix of $B$ we allocate a set $\Delta_{(i,i');(j,j')}$  which is obtained from subtracting each entry of  $\vec{\Delta}_{(i,i')j}$ from each element of  $\vec{\Delta}_{(i,i')j'}$ modulo $N$. Suppose $\left[\begin{array}{cc}
\left[x_1,y_1\right] & \left[x_2,y_2\right] \\
\left[x_3,y_3\right] & \left[x_4,y_4\right]
\end{array}\right]$ is a submatrix of $B$ in rows with indices $i$ and $i'$ and columns with indices $j$ and $j'$. Then,   $$\vec{\Delta}_{(i,i')j}=(x_1-x_3,x_1-y_3,y_1-x_3,y_1-y_3)$$ and $$\vec{\Delta}_{(i,i')j'}=(x_2-x_4,x_2-y_4,y_2-x_4,y_2-y_4).$$
The set $ \Delta_{(i,i');(j,j')}$ contains 16 elements.

{\small  $\begin{array}{lll}
\Delta_{(i,i');(j,j')}=\{&(x_1-x_3)-(x_2-x_4),& (x_1-x_3)-(x_2-y_4),\\ &(x_1-x_3)-(y_2-x_4),& (x_1-x_3)-(y_2-y_4),\\
&(x_1-y_3)-(x_2-x_4),&(x_1-y_3)-(x_2-y_4),\\
&(x_1-y_3)-(y_2-x_4),& (x_1-y_3)-(y_2-y_4)\\
&(y_1-x_3)-(x_2-x_4),& (y_1-x_3)-(x_2-y_4),\\
& (y_1-x_3)-(y_2-x_4),& (y_1-x_3)-(y_2-y_4),\\
&(y_1-y_3)-(x_2-x_4),& (y_1-y_3)-(x_2-y_4),\\
& (y_1-y_3)-(y_2-x_4),& (y_1-y_3)-(y_2-y_4)\}.\\
\end{array} $ }

\vspace{0.3cm}

We allocate every element of $ \Delta_{(i,i');(j,j')}$  to a 4-entry vector containing its entries. For example, $(x_1-x_3)-(x_2-x_4)\sim (x_1,x_3,x_2,x_4).$ If two  such vectors have two common components, then the equality of their corresponding elements in   $\Delta_{(i,i');(j,j')}$ causes a 4-cycle. For example,  $(x_1-x_3)-(x_2-x_4)$ and $(x_1-y_3)-(x_2-y_4)$ are two elements of $\Delta_{(i,i');(j,j')}$. Then, the two vectors corresponding to these elements are $(x_1,x_3,x_2,x_4)$ and $(x_1,y_3,x_2,y_4)$, respectively, which have two common elements $x_1,x_2$. The equality of these elements in $\Delta_{(i,i');(j,j')}$, that is $(x_1-x_3)-(x_2-x_4)=(x_1-y_3)-(x_2-y_4)$, implies the equality $-x_3+x_4=-y_3+y_4$ which according to Equation (\ref{Equation2}) gives a 4-cycle. If two vectors have one component in common, then the equality of their corresponding elements in  $\Delta_{(i,i');(j,j')}$ yields a 6-cycle. For example, from the equation $(y_1-x_3)-(x_2-y_4)=(y_1-y_3)-(y_2-x_4)$ we obtain $(y_3-x_3)-(x_2-y_4)+(y_2-x_4)=0$ which according to Equation (\ref{Equation2}) results in a 6-cycle. Hence, in a girth-8 Tanner graph,  all elements of $\Delta_{(i,i');(j,j')}$  whose corresponding vectors contain common coordinates  are distinct. The number of vectors with common coordinates is denoted by $\delta_{(i,i');(j,j')}$.
\begin{theorem}\label{ME-g8} 
 \cite{Sadeghi}	Given an $m\times n$ base matrix $W$ of a girth-8 multiple-edge QC-LDPC code, we define parameters ${\mathcal X'}, {\mathcal Y'}, {\mathcal Z'}$ and  ${\mathcal M}$:
	
	 	 $ {\mathcal X'} =\smash{\displaystyle\max_{0\leq i\leq m-1}}\{\sum_{j=0}^{n-1} W_{ij}: |W_{ij}|>1\}$,
	\vspace{0.2cm}
	
	  ${\mathcal Y'}=\smash{\displaystyle\max_{0\leq j\leq n-1}}\{\sum_{i=0}^{m-1}W_{ij}: |W_{ij}|>1\}$,
	\vspace{0.2cm}
	
		${\mathcal M}=\smash{\displaystyle\max_{0\leq i\neq i'\leq m-1}}\{\sum_{j=0}^{n-1}(W_{ij}+
		W_{i'j}):|W_{ij}|, |W_{i'j}|>1\}$,
	\vspace{0.2cm}
	
	 ${\mathcal Z'}=\max\{|\delta_{(i,i');(j,j')}|: 0\leq j,j'\leq n-1\}$,

\noindent	where if  ${\mathcal M}$ occurs in two rows $i,i'$, then ${\mathcal Z'}$ is computed on these two rows.  The lower bound on the lifting degree of a girth-8 multiple-edge QC-LDPC code is  $N=\max\{2{\mathcal X'},2{\mathcal Y'},{\mathcal M}+{\mathcal Z'}\}$.
\end{theorem}

The merit of our proposed base matrix obtained from the matrix $P$ can be seen in  the following lower bound on the lifting degree. 
	\begin{theorem}\label{TD-ME-g8}
	If $W$ is an $m\times n$ base matrix of a trade-based multiple-edge QC-LDPC code with girth 8, then $N\geq\max\{2{\mathcal X'},2{\mathcal Y'},{\mathcal M}\}$.	
	\end{theorem}
\begin{proof}
	Since the base matrix obtained from the matrix $P$ is free of $2\times2$ submatrices with all non-zero elements, the parameter $\delta_{(i,i');(j,j')}$ is zero for each $i,i'\in\{0,1,\dots,m-1\}$ and $j,j'\in\{0,1,\dots,n-1\}$. Therefore, in this case the value of ${\mathcal Z'}$ in Theorem \ref{ME-g8}  is zero.\hfill
	\end{proof}
	
\section{Trade-Based Time-Varying SC-LDPC Codes}\label{V}
 Let a $\emph{syndrome former memory order}$ (or syndrome memory) be $m_h$  and consider  $m\times n$ binary matrices $H_{i(j)}, i=0,1,\ldots,m_h$ and $j\geq0$. The parity-check matrix of an SC-LDPC code with constraint length $v_h=(m_h+1)n$ is defined as:
\begingroup\fontsize{8.5pt}{11pt}\begin{align}\label{SC-H}
H_{[\infty]}=\left[\begin{array}{ccccc}
 	H_{0(1)} & {\bf 0} & \ddots & {\bf 0}  & \ddots \\
 	H_{1(1)} & H_{0(2)} & \ddots & {\bf 0} & \ddots \\
 	\vdots & H_{1(2)} & \ddots & {\bf 0} &  \ddots\\
 	H_{m_h(1)} & \vdots  & \ddots & H_{0(t)} &\ddots \\
 	{\bf 0} & H_{m_h(2)}& \ddots &  H_{1(t)} & \ddots\\
 	{\bf 0} & {\bf 0} & \ddots & \vdots  & \ddots\\
 	{\bf 0} & {\bf 0} & \ddots   & H_{m_h(t)} & \ddots\\
 	{\bf 0}  & {\bf 0} & \ddots & {\bf 0}  & \ddots\\
 	\vdots &\vdots & \ddots & \vdots &  \ddots\\
 	\end{array}\right], 
\end{align}\endgroup 
 	{\em where ${\bf 0}$ is an $m\times n$ zero matrix. The syndrome matrices of the $t$-th column of $H_{[\infty]}$ in (\ref{SC-H})  give a matrix $H(t)=[H_0(t)|H_1(t)|\cdots|H_{m_h}(t)]^T$ which is the  matrix at time $t$. } According to the syndrome matrices $H_i(t)$, there are three types of SC-LDPC codes. If $H_i(t)=H_i(t+\tau)$ for each $0\leq i\leq m_h$ and for some $\tau>1$, then the parity-check matrix yields a $\emph{periodic}$ SC-LDPC code. The code is called $\emph{time-invariant}$ SC-LDPC if $H_i(t)=H_i(t+1)$. In other words, if the time dependence is dropped from the notation in  Equation (\ref{SC-H}). Otherwise, $H_{[\infty]}$ in (\ref{SC-H}) gives a $\emph{time-varying}$ SC-LDPC code.

 A $\emph{terminated}$ SC-LDPC code is defined by the parity-check matrix $H_{[L]}$,
\begingroup\fontsize{8.5pt}{11pt}\begin{align}\label{SC-HL}
H_{[L]}=\left[\begin{array}{cccc}
 	H_{0(1)} & {\bf 0}  & \ddots & {\bf 0}  \\
 	H_{1(1)} & H_{0(2)} & \ddots & {\bf 0} \\
 	\vdots & H_{1(2)} & \ddots & {\bf 0}\\
 	H_{m_h(1)}& \vdots  & \ddots & H_{0(L)} \\
 	{\bf 0} & H_{m_h(2)}& \ddots & H_{1(L)} \\
 {\bf 0}	 & {\bf 0} & \ddots & \vdots\\
 	{\bf 0} & {\bf 0} & \ddots  & H_{m_h(L)} \\
 	\end{array}\right]. 
\end{align}\endgroup 

  If all the columns of $H_{[L]}$ in (\ref{SC-HL}) have the same column weight, then the code is regular in its columns; otherwise, it is irregular.

In this section, we propose a novel approach to construct  time-varying SC-LDPC codes whose parity-check matrices are based on the trades of super-simple directed designs. In the following, we propose our method.

\begin{enumerate} 
\item {\em Consider a super-simple design yielding a matrix $C$ with $(m_h+1)m$ rows and $n$ columns. We divide $C$ into $m_h+1$ submatrices of size $m\times n$ named as syndrome matrices of $H(1)$.}
\item To define the syndrome matrices of the $i$-th column of $H_{[\infty]}$ in (\ref{SC-H}), that is $H(t)$, where $2\leq t\leq m_h$, we execute the following steps. We put $H_{0(t)}=H_{(t-1)(1)}$. Then, by removing all rows of $C$ appearing in $H_{0(t)}$, we divide the remaining $m_hm$ rows of $C$ to construct $m_h$ submatrices which yield the other $m_h$ syndrome matrices of $H(t)$.
\item  If a block design is constructed using base blocks, then each base block whose blocks appear in trades are associated to a syndrome matrix of $H(1)$. Then, for each time $t\neq1$, $H_{0(t)}$ is $H_{(t-1)(1)}$, and $H_{k(t)}$  is a rearranged form of  $H_{\ell(1)}$ in which $k>0, \ell\neq t$ and $H_{k(t)}$ is obtained by rearranging the row indices of  $H_{\ell(1)}$.  
\end{enumerate}

\begin{example}\label{EX1SC}
Using the matrix $C$ obtained in Example \ref{Ex2} we define the following $H_{i(t)}$ matrices.

{\tiny{$$\begin{array}{ccc}
H_{0(1)}=\left[\begin{array}{cccccccc}
1& 1 & . & . & . & . & . & .\\
. & . & 1& 1 & . & . & . & .\\
. & . & . & . & 1& 1 & . & .\\
. & . & . & . & . & . & 1 & 1
\end{array}\right] & H_{0(2)}=\left[\begin{array}{cccccccc}
1& 1 & . & . & . & . & . & .\\
. & . & 1& 1 & . & . & . & .\\
. & . & . & . & 1& 1 & . & .\\
. & . & . & . & . & . & 1 & 1
\end{array}\right] & H_{0(3)}=\left[\begin{array}{cccccccc}
. & . & . & . & 1 & . & 1 & .\\
. & . & . & 1& . & .  & . & 1\\
. & 1 & . & . & . & 1 & . & .\\
1 & . & 1 & . & . & . & . & .
\end{array}\right]\\
\\
 H_{1(1)}=\left[\begin{array}{cccccccc}
. & . & . & . & 1 & . & 1 & .\\
. & . & . & 1& . & .  & . & 1\\
. & 1 & . & . & . & 1 & . & .\\
1 & . & 1 & . & . & . & . & .
\end{array}\right] 
  & H_{1(2)}=\left[\begin{array}{cccccccc}
. & . & . & 1 & . & . & . & 1\\
. & . & 1 & . & 1 & .  & . & .\\
. & . & . & 1 & . & 1 & . & .\\
. & 1 & . & . & . & 1 & . & .
\end{array}\right] & H_{1(3)}=\left[\begin{array}{cccccccc}
1 & 1 & . & . & . & . & . & .\\
. & 1 & . & . & . & .  & 1 & .\\
. & . & . & . & . & . & 1 & 1\\
. & . & 1 & 1 & . & . & . & .
\end{array}\right] \\
\\
H_{2(1)}=\left[\begin{array}{cccccccc}
. & 1 & . & . & . & . & 1 & .\\
. & . & . & 1 & . & 1  & . & .\\
. & . & 1 & . & 1 & . & . & .\\
1 & . & . & . & . & . & . & 1
\end{array}\right] & H_{2(2)}=\left[\begin{array}{cccccccc}
. & . & . & . & 1 & . & 1 & .\\
1 & . & 1 & . & . & .  & . & .\\
. & 1 & . & . & . & . & 1 & .\\
1 & . & . & . & . & . & . & 1
\end{array}\right] 
 & H_{2(3)}=\left[\begin{array}{cccccccc}
. & . & . & . & 1 & 1 & . & .\\
. & . & . & 1 & . & 1  & . & .\\
. & . & 1 & . & 1 & . & . & .\\
1 & . & . & . & . & . & . &1
\end{array}\right] \\
\end{array}$$}}
 In the above matrices we have $H_{0(2)}=H_{0(1)}$ and $H_{0(3)}=H_{1(1)}$. Now we obtain the following parity-check matrix $H_{[3]}$ which gives a terminated time-varying SC-LDPC code with $L=3$, constraint length $24$ and column weight 3: 
 
 \begingroup\fontsize{8.5pt}{11pt}\begin{align}\label{EXS1-H}
H_{[3]}=\left[\begin{array}{ccc}
 	H_{0(1)} & 0 & 0 \\
 	H_{1(1)} & H_{0(2)} & 0 \\
 	H_{2(1)} & H_{1(2)}& H_{0(3)} \\
 	0& H_{2(2)}& H_{1(3)} \\
 	0& 0 & H_{2(3)} \\
 	\end{array}\right]. 
\end{align}\endgroup 
\end{example}

 Since the rows $H(1)$, $H(2)$ and $H(3)$ are columns of $C^T$ in Equation (\ref{DGDD28}),  instead of proposing the whole matrix and/or all syndrome matrices, we can provide column indices of $C^T$   which appear in each time of the parity-check matrix. For example, for $t=1$ of $H_{[3]}$ in (\ref{EXS1-H}), the first row of $H(1)$ is the last column of $C^T$ with column index $57$, the second row of $H(1)$ is the $11$-th column of $C^T$ with column index $56$ and continuing this process the last row of $H(1)$ is the first column of $C^T$ with column index $02$. Thus, $H(1)$ is a permutation of  the columns of $C^T$ with indices $[57,56,47,46,34,26,17,05,13,12,03,02]$. For $t=2$, the rows of $H(2)$ are the columns of $C^T$ with indices $[57,56,47,46,26,03,12, 17,34,05,13,02]$. For $t=3$, the rows of $H(3)$ are the columns of $C^T$ with indices $[34,26,17,05,57,13,46,56,47,12,03,02]$.

In the following, we point out properties of the parity-check matrix obtained from our method which gives a time-varying SC-LDPC code with girth at least 6.  

\begin{lemma}\label{lem1}
A necessary and sufficient condition to have a trade-based time-varying code with girth at least 6 is the non-existence of 4-cycles in the submatrices of the form $[H_{i(j)}|H_{i'(j')}]$, where $i'>i, j'<j$ and submatrices
$\left[\begin{array}{cc}
H_{i_1(j)} & H_{i_2(j')}\\
H_{i_3(j)} & H_{i_4(j')}\\
\end{array}\right]$, where $i_2<i_1,\ i_4<i_3$ and $j<j'$. 
\end{lemma}
\begin{proof}
Since each column of the matrix $H$ in (\ref{SC-H}) is obtained using the matrix $C$ which is free of 4-cycles, each column of $H$ is 4-cycle free. Moreover, since the first syndrome matrix $H_{0(j)}$ of the $j$-th column of $H$ is the same as $H_{(j-1)(0)}$, we do not need to check the existence of 4-cycles between $H_{0(j)}$ and $H_{j(0)}$. Thus, to have a time-varying code with girth at least 6, we should define each $H_{i(j)}$ such that $[H_{i(j)}|H_{i'(j')}]$ does not yield a 4-cycle, where $i'>i$ and $j'<j$. Moreover, we should propose syndrome matrices in a way that it guarantees the non-existence  4-cycles in each submatrix of the form 
$\left[\begin{array}{cc}
H_{i_1(j)} & H_{i_2(j')}\\
H_{i_3(j)} & H_{i_4(j')}\\
\end{array}\right]$, where $i_2<i_1,\ i_4<i_3$ and $j<j'$.\hfill
\end{proof}

According to Lemma \ref{lem1}, since each $H(t)$ is free of 4-cycles, checking 4-cycles in the parity-check matrix $H$ in (\ref{SC-H}) is limited to checking these cycles in each two $H(j)$ and $H(j')$, $j<j'$. This property holds for checking other $2k$-cycles in $H$. In other words, if the trade-based matrix is obtained from a directed design free of cyclical trades of volume $k$, then  each column of $H$ in (\ref{SC-H}) is $2k$-cycle free. 

 The lower bounds on $m_h$ and the constraint length of time-invariant SC-LDPC codes with girth 6, column weight $m$ and syndrome matrices of size $m\times n$ are $\lceil\frac{n-1}{2}\rceil$ and $n\left(\lceil\frac{n-1}{2}\rceil+1\right)$, respectively \cite{Battag1}. As shown in Lemma \ref{Lem2-SC}, we provide a lower bound on these parameters for time-varying SC-LDPC codes with girth 6.
 \begin{lemma}\label{Lem2-SC}
 Consider a directed super-simple design with $|{\mathcal{B}}|$ blocks in which the number of $(x_i,x_j)$ pairs appearing in trades are more than $|{\mathcal{B}}|$. Then, a necessary condition to construct a trade-based time-varying SC-LDPC code with girth 6 and syndrome matrices of size $m\times n$ are $n=|{\mathcal{B}}|$ and $(m_h+1)m\geq|{\mathcal{B}}|$. 
 \end{lemma}
 \begin{proof}
 Since the defining set of the directed designs we use to construct the matrix $C$ in Section \ref{III} have at least half of the blocks, all blocks are involved in the trade designs and thus the number of rows of $C$ is $|{\mathcal{B}}|$. Hence, $n=|{\mathcal{B}}|$. Moreover, since in the process of constructing the parity-check matrices we choose the directed designs so that the number of pairs ocurring in the trade designs are more than the number of blocks, the number of columns of $C$ is at least $|{\mathcal{B}}|$ and thus $m_h$ must  satisfy  the inequality $(m_h+1)m\geq|{\mathcal{B}}|$. \hfill
 \end{proof}
 
 \begin{example}\label{Ex2-col4}
 In Example \ref{Ex-col4}, from the $(20,5,1)$-DD we obtain a $64\times32$ matrix $C$ whose  transpose gives a $(2,4)$-regular LDPC code with girth 8. Now, we use the matrix $C$ to construct a trade-based time-varying SC-LDPC code with column weight 4 and syndrome matrices of size $m\times n$. First we identify the values of $m,n$ and $m_h$. It is clear that $n$,  the number of blocks of the design which appear in trades, is $|{\mathcal{B}}|=32$.  In order to define $H_{i(0)}$ for $i=0,\ldots,m_h$, we divide $C$ into $m_h+1$ submatrices. If we consider $m=8$, then according to Lemma \ref{Lem2-SC}, $(m_h+1)m\geq 32$. In other words, $m_h\geq 3$. Since, the number of rows in $C$, or the number of pairs in trades is 64, we can take $m_h=\frac{64}{m}=8$. Therefore, using $C$ we construct a trade-based time-varying SC-LDPC code with column weight 4.  The row weight of $H_s$ is 16 and the constraint length is $v_h=(m_h+1)|{\mathcal{B}}|=9\times32=288$. Similarly, if we take $m=4,16$, then $m_h=16,4$ with constraint length $v_h=544,160$, respectively.
 \end{example}
 
 Now we compare the syndrome memory and the constraint length of time-varying SC-LDPC code in Example \ref{EX1SC} with its counterpart time-invariant SC-LDPC codes with column weight 3, row weight 6 and girth 6 in \cite{Battag1}. For the latter case with syndrome matrices of size $3\times 6$ the lower bound on $m_h$ is 3 and the lower bound on the constraint length of the code is 24. Whereas, using our method the constraint length of the time-varying SC-LDPC code that we obtained with the same degree distribution and girth 6 is 24 with $m_h=2$.  Therefore, using our method we construct  the parity-check matrix of time-varying SC-LDPC codes whose constraint length meet the lower bound on the constraint length of time-invariant SC-LDPC code but with a smaller memory. Moreover, our proposed parity-check matrix does not rely on protographs and their associated exponent matrix.

\section{Conclusion}\label{VI}
In this paper, three studies are provided. First, we use the concept of trades in directed super-simple designs to present a new approach to construct LDPC codes with different column weights and row weights whose Tanner graphs have diverse girths. We provide some examples to show our method results in regular and irregular LDPC codes with different girths.  Second, we use the trade-based parity-check matrices to provide base matrices of multiple-edge QC-LDPC codes with girths 6 and 8. We show that these base matrices result in QC-LDPC codes in which lower bounds on the lifting degree are less than the ones in the literature. Moreover, the numerical results prove that  with our approach  the computational complexity of the search algorithm to define the exponent matrix is much less than other methods in the literature.  Third, we develop time-varying SC-LDPC codes whose parity-check matrices do not depend on protographs. In fact, using the trade-based matrices from the first point above we provide a method to directly construct the parity-check matrix of time-varying SC-LDPC codes.


\begin{thebibliography}{}
\bibitem{Gallager}
R. G. Gallager,  ``Low-Density Parity-Check Codes,"  {\it Cambridge, MA}, MIT Press, (1963).

\bibitem{Tanner}
R. M. Tanner,  ``A recursive approach to low complexity codes," {\it IEEE Trans. Inf. Theory}, vol. 27, no. 9, pp. 533--547, (1981).

\bibitem{Lin2014}
 J. Li, K. Liu, S. Lin, and K. Abdel-Ghaffar, ``Algebraic quasi-cyclic LDPC
codes: Construction, low error-floor, large girth and a reduced-complexity
decoding scheme,'' {\it IEEE Trans. Commun.}, vol. 62, no. 8, pp. 2626--2637,
(2014).

\bibitem{2004}
M. P. C. Fossorier,  ``Quasi-Cyclic Low-Density Parity-Check codes from circulant permutation matrices," {\it IEEE Trans. Inf. Theory}, vol. 50, no. 8, pp. 1788--1793, (2004).

\bibitem{Lentmaier2010}
 M.  Lentmaier, A. Sridharan, D. J. Costello, and K. S. Zigangirov, ``Iterative decoding threshold analysis for LDPC convolutional codes,"
{\it IEEE Trans. Inf. Theory}, vol. 56, no. 10, pp. 5274--5289, (2010).

\bibitem{Multiple}
H. Park, S. Hong, J. Seon and D. J. Shin,  ``Design of multiple-edge protographs for QC-LDPC codes avoiding short inevitable cycles," {\it IEEE Trans. Inf. Theory}, vol. 59,  no. 7, pp. 4598--4614, (2013).

\bibitem{Sadeghi}
M.-R Sadeghi,  ``Optimal Search for Girth-8  Quasi Cyclic and Spatially Coupled Multiple-Edge LDPC Codes," {\it IEEE Commun. Letters}, vol. 23, no. 9, pp. 1466--1469, (2019). 
\bibitem{LCOMM}
M.-R Sadeghi and F. Amirzade,  ``Analytical Lower Bound on the Lifting Degree of Multiple-Edge QC-LDPC Codes with Girth 6," {\it IEEE Commun. Letters}, vol. 22, no. 8, pp. 1528--1531, (2018).

\bibitem{Battag1}
M. Battaglioni, A. Tasdighi, G. Cancellieri, F. Chiaraluce,
and M. Baldi,  ``Design and Analysis of Time-Invariant
SC-LDPC Convolutional codes With Small Constraint
Length," {\it IEEE Trans. Commun}, vol. 66, no. 3, pp. 918--931, (2018).


\bibitem{Battag2}
M. H. Tadayon, A. Tasdighi, M. Battaglioni, M. Baldi
and F. Chiaraluce,  ``Efficient search of compact QC-LDPC
and SC-LDPC convolutional codes with large girth," {\it IEEE Commun. Letters}, vol. 22, no. 6, pp. 1156--1159, (2018).

\bibitem{Liu2015}
 Y. Liu, Y. Li, and Y. Chi, ``Spatially coupled LDPC codes constructed by parallelly connecting multiple chains," {\it IEEE Commun. Lett.}, vol. 19, no. 9, pp. 1472–1475, (2015).
 
 \bibitem{Truhachev}
D. Truhachev, D. G. M. Mitchell, M. Lentmaier,  D. J. Costello, Jr. and A. Karami ``Code Design Based on Connecting Spatially Coupled Graph Chains", {\it IEEE Trans. Inf. Theory}, vol. 65, no. 9, pp. 5604--5617, (2019).

\bibitem{9}
L. Dengsheng, L. Qiang and L. Shaoqian,  ``Construction of nonsystematic Low-Density Parity-Check Codes based on Symmetric Balanced Incomplete Block Designs," {\it J. Electron.}, vol. 25, no. 4, pp. 1--5, (2008).

\bibitem{10}
A. Gruner  and M. Huber, ``Low-Density Parity-Check Codes from Transversal  Designs with Improved Stopping Set Distribution," {\it IEEE Trans. Commun.},  vol. 61, no. 6, pp. 2190--2200, (2013).
\bibitem{11}
L. Lan, Y. Y. Tai, S. Lin, B. Memari and B. Honary,  ``New construction of quasi-cyclic LDPC codes based on special classes of BIBDs for the AWGN and binary erasure codes," {\it IEEE Trans. Commun.}, vol. 56, no. 1, pp. 39--48, (2008).

\bibitem{12}
B. Vasic and O. Milenkovic,  ``Combinatorial construction of low-density parity-check codes for iterative decoding," {\it IEEE Trans. Inf. Theory}, vol. 50, no. 6, pp. 1156--1176, (2004).

\bibitem{AMC}
F. Amirzade, M.-R. Sadeghi and D. Panario, ``QC-LDPC construction free of small size elementary trapping sets based on multiplicative subgroups of a finite field,” {\it Adv. Math. Commun.}, vol. 14, no. 3, pp. 397--411, (2020).

\bibitem{CRC}
C. J. Colbourn and J. H. Dinitz,  ``The CRC Handbook of Combinatorial Designs," Boca Raton FL, USA: CRC Press, (1996).

\bibitem{ChannelCodes}
W. E. Ryan and S. Lin, ``Channel Codes Classical and Modern," {\it Cambridge University Press}, (2009).

\bibitem{street}
E. S. Mahmoodian, N. Soltankhah,  ``On defining sets of directed designs," {\it Australas. J. Combin.}, vol. 19, pp. 179--190, (1999).

\bibitem{Soltan}
N. Soltankhah,  ``On directed trades," {\it Australas. J. Combin.}, vol. 11, pp. 59--66, (1995).

\bibitem{mullin}
R. C. Mullin and H.-D. O. F. Gronau,  ``PBDs and GDDs: the basics," in {\it The CRC 
Handbook of Combinatorial Designs, second edition (ed. C. J. Colbourn and J. H. 
Dinitz), CRC Press}, pp. 231--236, (2007).	

\bibitem{Quinn}
M. J. Grannell, T. S. Griggs and K. A. S. Quinn, ``Smallest defining sets of 
directed triple systems," {\it Discrete Math.}, vol. 309, pp.  4810--4818, 2009.

\bibitem{farzane1}
F. Amirzade and N. Soltankhah, ``Smallest defining sets of super-simple 2-$(v,4,1)$ 
directed designs", {\it Utilitas Math.}, vol. 96, pp. 331--344, (2015).

\bibitem{Boostan}
M. Boostan, S. Golalizadeh and N. Soltankhahh, ``Super-simple 2-$(v, 4, 2)$ directed 
designs and lower bound for the minimum size of their defining set," {\it Discrete 
Appl. Math.}, vol. 201, pp. 14--23, (2016).

\bibitem{farzane2}
F. Amirzade and N. Soltankhah, ``On super-simple 2-$(v, 5, 1)$ directed designs 
and their smallest defining sets," {\it Australas. J. Combin.}, vol. 54, pp. 
85--106, (2012).

\end{thebibliography}
\end{document}